\documentclass[conference]{IEEEtran}
\IEEEoverridecommandlockouts
\usepackage{tikz, tikzpagenodes, dblfloatfix, multirow}
\usetikzlibrary{calc, positioning, shapes, decorations.pathreplacing}

\usepackage{amssymb}
\usepackage{amsfonts}
\usepackage{cite}
\usepackage{amsmath,amssymb,amsfonts,amsthm}
\newtheorem{theorem}{Theorem}
\newtheorem{definition}{Definition}
\usepackage{algorithmic}
\usepackage{graphicx}
\usepackage{textcomp}
\usepackage{xcolor}
\usepackage{booktabs}
\usepackage{multirow}
\usepackage{makecell} % for \makecell
\usepackage{siunitx} % optional, for nice number alignment
\usepackage{enumitem}
\usepackage{subcaption}
\usepackage{hyperref}
\usepackage{listings}
\usepackage{xcolor}

\definecolor{myblue}{RGB}{180, 180, 180}
\definecolor{mylightblue}{RGB}{220, 220, 220}
\definecolor{myred}{RGB}{231, 76, 60}
\definecolor{myotherred}{RGB}{192, 57, 43}
\definecolor{myyellow}{RGB}{135, 206, 250}
\definecolor{myorange}{RGB}{230, 126, 34}
\def\BibTeX{{\rm B\kern-.05em{\sc i\kern-.025em b}\kern-.08em
    T\kern-.1667em\lower.7ex\hbox{E}\kern-.125emX}}

\begin{document}

\title{Wavefront Parallelization for Efficient Learned Image Compression\thanks{This paper is supported by the MIC Project for Efficient Frequency Utilization Toward Wireless IP Multicasting and JST CRONOS, Japan Grant Number JPMJCS25N2.}}

\author{\IEEEauthorblockN{Shimon Murai}
\IEEEauthorblockA{\textit{School of Fundamental} \\
\textit{Science and Engineering} \\
\textit{Waseda University}\\
Tokyo, Japan \\
octachoron@suou.waseda.jp}
% rewrite that fucking shit, it's lame and confusing
\and
\IEEEauthorblockN{Fangzheng Lin}
\IEEEauthorblockA{\textit{School of Engineering} \\
\textit{Institute of Science Tokyo}\\
Tokyo, Japan \\
lin.f.f849@m.isct.ac.jp}
\and
\IEEEauthorblockN{Kasidis Arunruangsirilert}
\IEEEauthorblockA{
\textit{School of Fundamental} \\
\textit{Science and Engineering} \\
\textit{Waseda University}\\
Tokyo, Japan \\
kasidis@katto.comm.waseda.ac.jp}
\and
\IEEEauthorblockN{Jiro Katto}
\IEEEauthorblockA{
\textit{School of Fundamental} \\
\textit{Science and Engineering} \\
\textit{Waseda University}\\
Tokyo, Japan \\
katto@waseda.jp}}
  %\ninept
  %
  \maketitle
  \begin{abstract}
    Autoregressive context models are foundational for learned image compression,
    but they suffer from slow serial inference. Existing acceleration methods such
    as checkerboard context require architectural changes and retraining, thus
    are inapplicable to pre-trained models. We propose a completely training-free inference-time
    acceleration algorithm inspired by wavefront parallelism in video
    coding standards. Our method reorganizes inference into an optimal ``staggered''
    wavefront order, minimizing sequential steps while maintaining exact autoregressive
    dependencies. Experimental results show our approach accelerates pre-trained
    autoregressive models (e.g., Cheng et al.) by more than $13\times$ while preserving
    the original rate-distortion performance. We also demonstrate that faster
    decoding is possible by trading off precise context dependencies. Source code
    will be available at \url{https://github.com/tokkiwa/compressai-wavefront}.
  \end{abstract}
  \vspace{-1mm}
  \section{Introduction}
  \label{sec:intro} Learned Image Compression (LIC) has surpassed traditional codecs
  in rate-distortion (RD) performance thanks to autoregressive context models~\cite{minnen2018,
  cheng2020, heELICEfficientLearned2022} which effectively capture spatial correlations.
  However, the sequential dependency of such models requires a raster-scan decoding
  order, lacking parallelism as decoding requires a minimum of $O(H \times W)$
  stages~\cite{minnen2018,cheng2020}. To accelerate inference, recent works
  introduce parallelism by breaking the latent-by-latent sequential dependency, such
  as using Checkerboard patterns\cite{heCheckerboardContextModel2021b} or channel-wise autoregression\cite{minnen2020}.
  However, these methods fundamentally alter the probabilistic model, so models not
  only have to be retrained but also typically suffer from degraded RD performance.
    In contrast, spatial context models remain relevant in RD-oriented designs. For example, recent high-RD lattice-quantized codecs still rely on spatial autoregressive entropy modeling~\cite{kudo2023lvqvae,lei2025ltc}, while achieving RD performance competitive with recent codecs based on alternative context models.
     This motivates accelerating pretrained spatial AR context models as a complementary systems contribution. Therefore, achieving the computation speed of parallel decoding without sacrificing the
  high RD performance of autoregressive models remains an open challenge. \looseness=-2

  \begin{figure}[t]
    \setlength{\belowcaptionskip}{-10pt}
    \centering
    \includegraphics[width=\linewidth]{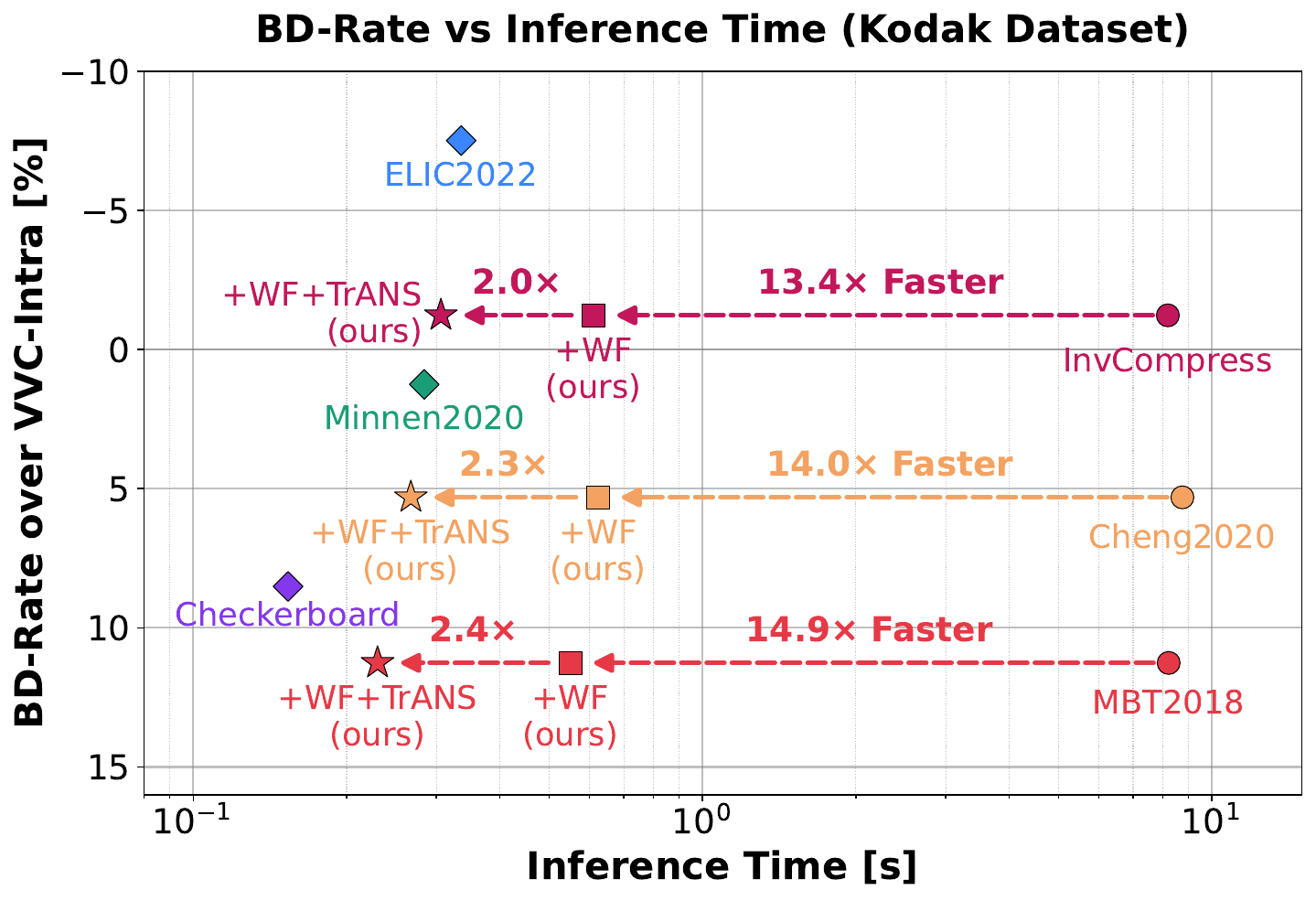}
    \caption{BD-rate versus inference time (compress and decompress) on Kodak
    dataset\cite{kodak}. We apply our Wavefront (``WF'') parallelization to
    three baselines: MBT2018\cite{minnen2018}, Cheng2020\cite{cheng2020} and
    InvCompress\cite{xie_enhanced_inverible_2021}. ``TrANS'' denotes Tensorized
    rANS, an optimized entropy coder with tensorized interface (Sec.~\ref{sec:wavefront_impl}).}
    \label{fig:bdrate_time_kodak}
  \end{figure}

  \begin{figure*}[btp]
\setlength{\belowcaptionskip}{-2pt} 
    \centerline{\tikzset{font={\fontsize{8pt}{12}\selectfont}}

\newcommand{\rulesep}{\unskip\ \vrule\ }

\begin{subfigure}
    [t]{0.23\textwidth}
    \centering
    \begin{tikzpicture}[x=0.4cm, y=0.4cm, step=0.4cm]
        \draw[fill=myblue] (0,5) rectangle (8,8) (0,4) rectangle (3,6);
        \draw[fill=myred] (4,4) rectangle (5,5);
        \fill[myyellow] (2,5) rectangle (7,7) (2,4) rectangle (4,5);
        \draw (0,0) grid (8,8);
        \draw[thick] (0,0) rectangle (8,8);
        \draw[dashed, ultra thick] (2,2) rectangle (7,7);

        \draw[draw=black, line width=0.3mm, -stealth] (5,4.5) -- (5.5,4.5);
    \end{tikzpicture}
    \caption{Autoregressive (AR) \cite{cheng2020}}
    \label{fig:related_work:ar}
\end{subfigure}\hfill
\begin{subfigure}
    [t]{0.23\textwidth}
    \centering
    \begin{tikzpicture}[x=0.4cm, y=0.4cm, step=0.4cm]
        \foreach \y in {0,2,...,6}{ \foreach \x in {0,2,...,6}{ \fill[myblue] (\x,1+\y) rectangle (1+\x,2+\y) (1+\x,\y) rectangle (2+\x,1+\y); \fill[myred] (\x,\y) rectangle (1+\x,1+\y) (1+\x,1+\y) rectangle (2+\x,2+\y); } }

        \fill[myotherred!20] (2,2) rectangle +(5,5);
        \fill[myyellow] (3,2) rectangle (4,3) (5,2) rectangle (6,3);
        \fill[myyellow] (3,4) rectangle (4,5) (5,4) rectangle (6,5);
        \fill[myyellow] (3,6) rectangle (4,7) (5,6) rectangle (6,7);
        \fill[myyellow]
            (2,3) rectangle
            (3,4)
            (4,3) rectangle
            (5,4)
            (6,3) rectangle (7,4);
        \fill[myyellow]
            (2,5) rectangle
            (3,6)
            (4,5) rectangle
            (5,6)
            (6,5) rectangle (7,6);

        \foreach \y in {0,2,...,6}{ \foreach \x in {0,2,...,6}{ \node[white] at (\x + 0.5, \y + 0.5) {1}; \node[white] at (\x + 1.5, \y + 1.5) {1}; \node at (\x + 0.5, \y + 1.5) {0}; \node at (\x + 1.5, \y + 0.5) {0}; } }
        \foreach \y in {2,4,6}{ \foreach \x in {2,4,6}{ \node at (\x + 0.5, \y + 0.5) {1}; } }
        \foreach \y in {3,5}{ \foreach \x in {3,5}{ \node at (\x + 0.5, \y + 0.5) {1}; } }

        \fill[myred] (4, 4) rectangle (5,5);
        \node[white] at (4.5, 4.5) {1};

        \draw (0,0) grid (8,8);
        \draw[dashed, ultra thick] (2,2) rectangle (7,7);
        \draw[thick] (0,0) rectangle (8,8);
    \end{tikzpicture}
    \caption{Checkerboard \cite{heCheckerboardContextModel2021b}}
    \label{fig:related_work:checkerboard}
\end{subfigure}\hfill
\begin{subfigure}
    [t]{0.23\textwidth}
    \centering
    \begin{tikzpicture}[x=0.4cm, y=0.4cm, step=0.4cm]
        % Wavefront lambda=2 (for K=3 kernel)
        % t = 2*(7-y) + x
        % Snapshot at t=10
        % Focus pixel: x=4, y=4 (which has t = 2*(7-4)+4 = 10)

        % 1. Background colors (Blue/Pink/Red/White)
        % Pink = same wavefront, Red = focus pixel only
        \foreach \y in {0,...,7}{ \foreach \x in {0,...,7}{ \pgfmathtruncatemacro{\w}{2*(7-\y)+\x} \ifnum \w < 10 \fill[myblue] (\x,\y) rectangle ++(1,1); \fi \ifnum \w = 10 \fill[myred] (\x,\y) rectangle ++(1,1); \fi } }
        \fill[myred] (4,4) rectangle ++(1,1);
        \fill[myyellow] (3,5) rectangle (6,6);
        \fill[myyellow] (3,4) rectangle (4,5);

        % 3. Grid and numbers
        \draw (0,0) grid (8,8);
        \foreach \y in {0,...,7}{ \foreach \x in {0,...,7}{ \pgfmathtruncatemacro{\w}{2*(7-\y)+\x} \ifnum \w = 10 \node[white] at (\x+0.5, \y+0.5) {\w}; \else \node[black] at (\x+0.5, \y+0.5) {\w}; \fi } }

        % 4. Dashed Kernel Box (3x3)
        \draw[dashed, ultra thick] (3,3) rectangle (6,6);
        \draw[thick] (0,0) rectangle (8,8);
    \end{tikzpicture}
    \caption{Wavefront ($\lambda=2$)}
    \label{fig:related_work:wf2}
\end{subfigure}\hfill
\begin{subfigure}
    [t]{0.23\textwidth}
    \centering
    \begin{tikzpicture}[x=0.4cm, y=0.4cm, step=0.4cm]
        % Wavefront lambda=3 (for K=5 kernel)
        % t = 3*(7-y) + x
        % Snapshot at t=15
        % Focus pixel: x=3, y=3 (which has t=15)

        % 1. Background colors (Blue/Pink/Red/White)
        % Pink = same wavefront, Red = focus pixel only
        \foreach \y in {0,...,7}{ \foreach \x in {0,...,7}{ \pgfmathtruncatemacro{\w}{3*(7-\y)+\x} \ifnum \w < 15 \fill[myblue] (\x,\y) rectangle ++(1,1); \fi \ifnum \w = 15 \fill[myred] (\x,\y) rectangle ++(1,1); \fi } }
        \fill[myred] (3,3) rectangle ++(1,1);
        \fill[myyellow] (1,4) rectangle (6,6);
        \fill[myyellow] (1,3) rectangle (3,4);

        % 3. Grid and numbers
        \draw (0,0) grid (8,8);
        \foreach \y in {0,...,7}{ \foreach \x in {0,...,7}{ \pgfmathtruncatemacro{\w}{3*(7-\y)+\x} \ifnum \w = 15 \node[white] at (\x+0.5, \y+0.5) {\w}; \else \node[black] at (\x+0.5, \y+0.5) {\w}; \fi } }

        % 4. Dashed Kernel Box (5x5)
        \draw[dashed, ultra thick] (1,1) rectangle (6,6);
        \draw[thick] (0,0) rectangle (8,8);
    \end{tikzpicture}
    \caption{Wavefront ($\lambda=3$)}
    \label{fig:related_work:wf3}
\end{subfigure}}
    \centerline{\tikzset{font={\fontsize{9pt}{12}\selectfont}}

\begin{tikzpicture}[x=0.4cm, y=0.4cm, step=0.4cm]
    \draw[fill=white] (3,0) rectangle +(1,1) node (white0) {};
    \node[right=0 of white0.east, anchor=north west] {Not yet decoded};

    \draw[fill=myblue] (12,0) rectangle +(1,1) node (blue) {};
    \draw[fill=myyellow] (13.5,0) rectangle +(1,1) node (yellow0) {};
    \node[right=0 of yellow0.east, anchor=north west] {Decoded latent codes};

    \draw[fill=myred] (23,0) rectangle +(1,1) node (red) {};
    \draw[fill=myotherred!20] (24.5,0) rectangle +(1,1) node (otherred0) {};
    \node[right=0 of otherred0.east, anchor=north west] {Currently decoding};

    \draw[dashed, ultra thick] (3, -1.5) rectangle +(1,1) node (dashed) {};
    \node[right=0 of dashed.east, anchor=north west] {Context of conv kernel};

    \draw[fill=myyellow] (13.5,-1.5) rectangle +(1,1) node (yellow1) {};
    \draw[dashed, ultra thick] (13.5,-1.5) rectangle +(1,1);
    \node[right=0 of yellow1.east, anchor=north west] {is available as context};

    \draw[fill=white] (23,-1.5) rectangle +(1,1) node (white1) {};
    \draw[fill=myotherred!20] (24.5,-1.5) rectangle +(1,1) node (otherred1) {};
    \draw[dashed, ultra thick] (23,-1.5) rectangle +(1,1);
    \draw[dashed, ultra thick] (24.5,-1.5) rectangle +(1,1);
    \node[right=0 of otherred1.east, anchor=north west] {are not available};

    %\draw[dashed] (2.75, -0.25) rectangle (32.5, 1.25);
    \draw[dashed] (2.75, -1.75) rectangle (32.5, -0.25);

    \draw (33,-0.5) rectangle +(1,1);
    \draw (34.5,-0.5) rectangle +(1,1);
    \node at (33.5,0) {0};
    \node at (35,0) {9};
    \node[anchor=west] at (35.7,0) {Decoding steps (pixels with the};
    \node[anchor=north west, inner sep=0pt]
        at
        (33,-0.65)
        {same number can be processed in parallel)};
\end{tikzpicture}}
    \caption{Illustration of different context models; \subref{fig:related_work:wf2}
    \subref{fig:related_work:wf3} are examples of proposed staggered wavefront processing.}
    \label{f:related_work}
  \end{figure*}
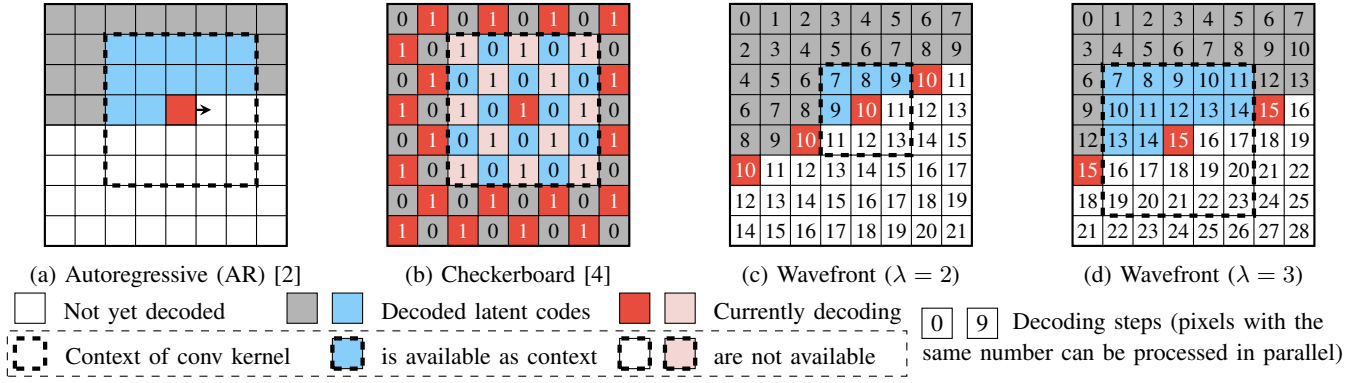
  In this paper, we revisit the context model as a Bayesian network defined on a
  directed acyclic graph. A crucial observation is that the dependency graph
  induced by spatial autoregression is sparse: a pixel depends only on a small local
  neighborhood (e.g., a $5 \times 5$ window) rather than the entire history of decoded
  pixels, implying that the standard raster-scan order is not necessarily
  optimal scheduling.
  % While wavefront parallelization has been successfully applied in video coding standards
  % (e.g., HEVC/VVC's Wavefront Parallel Processing), the pixel-wise dependencies
  % of masked convolutions in LIC are fundamentally different from block-wise
  % dependencies in video codecs, requiring rigorous mathematical analysis to derive
  % the optimal schedule.
  This leads to our proposed \textbf{staggered wavefront} strategy, which maximizes
  parallelism while strictly respecting causality. We adopt Lamport's theorem\cite{lamport1974parallel}
  from scheduling theory to prove that this order is a valid and optimal
  scheduling under certain restrictions. Moreover, we propose the \textbf{context
  approximation} mode to further increase parallelism by processing multiple wavefronts
  jointly and approximating unavailable context values. This mode allows a
  single trained model to operate at multiple speed (parallelism) settings, trading
  off bitrate against computation time. In contrast to prior acceleration
  methods, our proposal is \textbf{completely training-free}; it is a drop-in replacement
  for inference on pre-trained models to enable parallel execution. Our main
  contributions are:
  \begin{itemize}
    \item We redefine the autoregressive model as a directed acyclic graph and derive
      an optimal parallel scheduling using Lamport’s theorem, and propose a
      training-free parallelization method that achieves $\approx 13$--$32\times$
      speedup.% depending on image resolution.

    \item We demonstrate that by approximating spatial context, a single trained
      model can operate at levels of parallelism, providing a tradeoff between bitrate
      and computation time.

    \item We provide a GPU-friendly batched implementation of the proposed scheduling
      method and discuss its design and performance, achieving an additional
      acceleration of $\approx$ 2--3$\times$.

    \item We open-source our code to facilitate reproducibility.
  \end{itemize}

  \section{Background}

  \subsection{Learned Image Compression}
  Modern learned image compression frameworks typically follow  the structure of a variational
  autoencoder (VAE)\cite{kingma2014autoencoding}. An input image $x$ is transformed
  by an analysis transform $g_{a}$ into a latent representation $y = g_{a}(x)$.
  These latents are then quantized to yield discrete symbols $\hat{y}= Q(y)$, which
  are losslessly compressed into a bitstream with an entropy coder. During
  decoding, the bitstream is decompressed to recover $\hat{y}$, which is then
  passed through a synthesis transform $g_{s}$ to reconstruct the image
  $\hat{x}= g_{s}(\hat{y})$.

  The goal of the network is to minimize a rate-distortion objective function
  $\mathcal{L}= R + \lambda D$, where $D$ measures the distortion between $x$ and
  $\hat{x}$ (e.g., MSE or MS-SSIM), and $R$ represents the bitrate estimated
  with cross entropy. Therefore, accurately estimating the probability
  distribution of the quantized latents is crucial for minimizing $R$ and
  achieving high compression efficiency. \looseness=-2

  \subsection{Autoregressive Context Modeling}
  To minimize the bitrate, the entropy model must capture the statistical dependencies
  among latent symbols. Early methods employed a factorized prior\cite{balle2017},
  which assumes latent variables are independent and thus not as performant. The
  hyperprior model\cite{ballé2018variational} addresses this by utilizing side
  information $\hat{z}$ to predict the distribution parameters of $\hat{y}$.

  To further improve compression, recent methods incorporate autoregressive context
  models, which can be broadly categorized into two types: channel-wise autoregression\cite{minnen2020},
  capturing dependencies across channels; and spatial autoregression\cite{minnen2018},
  capturing dependencies among neighboring latent pixels within the same spatial
  plane. The general autoregressive model can be expressed as:
  \begin{equation}
    p(\hat{y}| \hat{z}) = \prod_{i}p(\hat{y}_{i}| \text{ctx}(\hat{y}_{i}), \hat{z}
    )
  \end{equation}
  While dense spatial autoregression shown in Fig.~\ref{f:related_work}\subref{fig:related_work:ar}
  offers powerful probability modeling capabilities, it also introduces a severe
  bottleneck: the decoding of a latent pixel strictly depends on its causal neighbors,
  and thus enforces a sequential raster-scan execution, leading to high latency.

  To mitigate this computational cost, recent performance-optimized models adopt
  simplified architectures that aim to reduce the number of serial stages. For example,
  Checkerboard models\cite{heCheckerboardContextModel2021b} reorganize spatial dependencies
  into a two-pass parallel decoding scheduling (Fig.~\ref{f:related_work}\subref{fig:related_work:checkerboard}).
  Although these approaches effectively accelerate inference, in the end they
  are architectural approximations of the probability model. Adopting these methods
  also requires full re-training of the LIC model. These reasons combined are
  why standard spatial autoregressive models\cite{minnen2018, cheng2020,
  xie_enhanced_inverible_2021} still remain attractive: they offer dense
  correlation modeling, and we argue that their practicality would highly improve
  when the aforementioned serial bottleneck is resolved without relying on architectural
  simplification.\looseness=-2

  \subsection{Wavefront Parallel Processing in Video Coding}
  Wavefront Parallel Processing (WPP) is a standardized methodology in HEVC and VVC
  that enables row-level parallelism during CTU decoding \cite{sze2014hevc,HEVCComplecity}.
  Each CTU row begins processing once a small number of CTUs (typically one or two)
  in the row above have been completed, creating a diagonal wavefront pattern.
  However, applying WPP directly to learned image compression is non-trivial.
  Video codecs operate at block granularity with localized inter-block
  dependencies, whereas spatial autoregressive learned image compression models operate
  with dense and complex receptive fields induced by masked convolutions. To the
  best of our knowledge, no prior work has formally derived the optimal wavefront
  schedule for pixel-wise autoregressive context models. \looseness=-1
  \section{Proposed Method}

  \subsection{Problem Formulation}
  Let $Y \in \mathbb{R}^{H \times W \times C}$ be the latent feature map to be
  compressed. In spatial autoregressive entropy models, the probability
  distribution of a latent pixel $y_{i,j}$ (at row $i$, column $j$) is conditioned
  on the set of previously decoded pixels, denoted as the context
  $\mathcal{C}_{i,j}$.
  \begin{equation}
    p(Y) = \prod_{i=0}^{H-1}\prod_{j=0}^{W-1}p(y_{i,j}\mid \mathcal{C}_{i,j})
  \end{equation}
  Typical implementations define $\mathcal{C}_{i,j}$ using a masked convolution with
  a kernel size $K \times K$. The causal context restricts dependencies to
  pixels that are spatially prior to $(i, j)$, leading to the decoding in raster-scan
  order. We show that better scheduling that fully preserves context exists
  while also allowing parallelism.

  We formally redefine this structure as a \textbf{Bayesian network} on a directed
  acyclic graph (DAG) $G=(V, E)$. The vertex set $V$ represents the spatial
  locations $\{(i,j)\}$, and the edge set $E$ represents the dependencies defined
  by the kernel:
  \begin{equation}
    E = \{ ((i', j') \to (i, j)) \mid (i-i', j-j') \in \mathcal{K}\}
  \end{equation}
  where $\mathcal{K}$ is the set of relative coordinates with non-zero weights
  in the causal kernel.

  For the decoding process to be valid and preserve all causal contexts, the execution
  order must be a \textbf{topological sort} of $G$. The standard raster-scan order,
  defined by the index $k = i \cdot W + j$, satisfies this condition. While
  sufficient for correctness, it imposes a strict total ordering that prohibits parallel
  execution. However, the graph $G$ is sparse; a pixel only depends on a local
  neighborhood. This implies that multiple nodes may share the same topological rank.
  Our goal is to find an optimal topological sorting function $\tau(i, j) \to t$
  that minimizes the total number of steps $T_{\max}= \max_{(i,j)}\tau(i,j)$
  while satisfying the causality constraint:
  \begin{equation}
    \forall ((i', j') \to (i, j)) \in E, \quad \tau(i', j') < \tau(i, j)
  \end{equation}
  Nodes with the same rank $t$ form a \textit{Wavefront} set
  $\mathcal{W}_{t}= \{ (i,j) \mid \tau(i,j) = t \}$ and are mutually d-separated,
  given $\bigcup_{k<t}\mathcal{W}_{k}$. Thus, all pixels in $\mathcal{W}_{t}$
  can be decoded in parallel.

  \subsection{Optimal Sorting via Hyperplane Method}
  \label{sec:optimal_sorting} To derive the optimal schedule $\tau$, we formulate
  this as a loop scheduling problem on a 2D grid. We employ \textbf{Lamport's
  Hyperplane Method} \cite{lamport1974parallel}, which seeks a linear schedule
  vector $\boldsymbol{\pi}= (\lambda , 1)$ such that the time step is given by:
  \begin{equation}
    \tau(i, j) = \boldsymbol{\pi}\cdot [i, j]^{T}= \lambda i + j
  \end{equation}
  Here, $\lambda$ is a parameter called shear factor. According to Lamport's theorem,
  this schedule is valid if and only if:
  \begin{equation}
    \boldsymbol{\pi}\cdot \mathbf{d}> 0, \quad \forall \mathbf{d}\in \mathcal{D}\label{eq:lamport_condition}
  \end{equation}
  where $\mathcal{D}$ is the set of dependence vectors $\mathbf{d}= (d_{i}, d_{j}
  )$ corresponding to the relative positions of the causal neighbors.

  For a standard $5 \times 5$ masked convolution, the dependency set includes
  vectors from the previous rows. The most critical constraint comes from the neighbor
  in the upper right in the receptive field, corresponding to the vector
  $\mathbf{d}_{\text{crit}}= (1, -2)$ (i.e., row $i-1$, column $j+2$).
  Substituting $\mathbf{d}_{\text{crit}}$ into Eq. \ref{eq:lamport_condition}:
  \begin{equation}
    \lambda \cdot 1 + 1 \cdot (-2) > 0 \implies \lambda > 2
  \end{equation}
  Since $\tau$ must assign integer time steps to grid points, $\lambda$ must also
  be an integer. Thus, the minimum valid shear factor is \textbf{$\lambda = 3$}. The
  number of execution steps of $\tau$ equals the length of the longest path of the
  DAG, implying its optimality among other schedules.

  More generality, for any
  $(2R + 1) \times (2R +1)$ masked convolution, the wavefront schedule with shear
  factor $\lambda = R + 1$ is valid and optimal. We depict this schedule, a \textbf{staggered wavefront}, in Fig.
  \ref{fig:related_work:wf2} and Fig. \ref{fig:related_work:wf3}. The standard
  wavefront used in HEVC/VVC (as detailed in Sec. 3.3.2. of \cite{sze2014hevc})
  %\cite{sze2014hevc}
  often corresponds to $\lambda=1$ (block-wise causality) or $\lambda=2$ in this
  staggered wavefront framework.

  \begin{figure}[t]
    \setlength{\belowcaptionskip}{-5pt}
    \centering
    \tikzset{font={\fontsize{8pt}{12}\selectfont}}

% Grouped decoding with context approximation (1-column figure)
% Left: Standard Wavefront (group_size=1)
% Right: Grouped Wavefront (group_size=2) with mean-fill approximation

\begin{subfigure}
    [t]{0.48\columnwidth}
    \centering
    \begin{tikzpicture}[x=0.4cm, y=0.4cm, step=0.4cm]
        % Standard Wavefront lambda=2, t=10 (single diagonal)
        % Focus pixel: (4,4)

        % 1. Background colors
        \foreach \y in {0,...,7}{ \foreach \x in {0,...,7}{ \pgfmathtruncatemacro{\w}{2*(7-\y)+\x} \ifnum \w < 10 \fill[myblue] (\x,\y) rectangle ++(1,1); \fi \ifnum \w = 10 \fill[myred] (\x,\y) rectangle ++(1,1); \fi } }
        % Focus pixel (4,4) is red
        \fill[myred] (4,4) rectangle ++(1,1);

        % 2. Yellow Context for pixel (4,4) with 3x3 kernel
        \fill[myyellow] (3,5) rectangle (6,6);
        \fill[myyellow] (3,4) rectangle (4,5);

        % 3. Grid and numbers
        \draw (0,0) grid (8,8);
        \foreach \y in {0,...,7}{ \foreach \x in {0,...,7}{ \pgfmathtruncatemacro{\w}{2*(7-\y)+\x} \ifnum \w = 10 \node[white] at (\x+0.5, \y+0.5) {\w}; \else \node[black] at (\x+0.5, \y+0.5) {\w}; \fi } }

        % 4. Dashed Kernel Box (3x3)
        \draw[dashed, ultra thick] (3,3) rectangle (6,6);
        \draw[thick] (0,0) rectangle (8,8);
    \end{tikzpicture}
    \caption{Standard ($N{=}1$)}
    \label{fig:grouping:standard}
\end{subfigure}\hfill
\begin{subfigure}
    [t]{0.48\columnwidth}
    \centering
    \begin{tikzpicture}[x=0.4cm, y=0.4cm, step=0.4cm]
        % Grouped Wavefront lambda=2, group_size=2
        % Processing group 5 (t=10,11) simultaneously
        % Focus pixel: (5,4) with t=11, group=5

        % 1. Background colors: Blue (decoded groups 0-4), Red (current group 5)
        % t = 2*(7-y)+x, group = floor(t/2)
        \foreach \y in {0,...,7}{ \foreach \x in {0,...,7}{ \pgfmathtruncatemacro{\w}{2*(7-\y)+\x} \pgfmathtruncatemacro{\grp}{int(\w/2)} \ifnum \grp < 5 \fill[myblue] (\x,\y) rectangle ++(1,1); \fi \ifnum \grp = 5 \fill[myred] (\x,\y) rectangle ++(1,1); \fi } }

        % Focus pixel (5,4) stays red
        \fill[myred] (5,4) rectangle ++(1,1);

        % 2. Context for pixel (5,4)
        % Yellow = exact context (already decoded, group < 5)
        % Pink = same group, unavailable context (mean-filled)
        % For 3x3 kernel around (5,4):
        % (4,5): t=8, g=4 -> decoded (yellow)
        % (5,5): t=9, g=4 -> decoded (yellow)
        % (6,5): t=10, g=5 -> SAME GROUP! (pink, approximated)
        % (4,4): t=10, g=5 -> SAME GROUP! (pink, approximated)
        % (6,4): t=12, g=6 -> future, not in kernel

        \fill[myyellow] (4,5) rectangle (6,6);
        \fill[myyellow] (4,4) rectangle (5,5);

        % Pink squares for same-group approximated context
        \fill[myred!30] (6,5) rectangle (7,6);
        \fill[myred!30] (4,4) rectangle (5,5);

        % 3. Grid and numbers (group numbers instead of t)
        \draw (0,0) grid (8,8);
        \foreach \y in {0,...,7}{ \foreach \x in {0,...,7}{ \pgfmathtruncatemacro{\w}{2*(7-\y)+\x} \pgfmathtruncatemacro{\grp}{int(\w/2)} \ifnum \grp = 5 \ifnum \x = 4 \ifnum \y = 4 \node[black] at (\x+0.5, \y+0.5) {\grp}; \else \node[white] at (\x+0.5, \y+0.5) {\grp}; \fi \else \ifnum \x = 6 \ifnum \y = 5 \node[black] at (\x+0.5, \y+0.5) {\grp}; \else \node[white] at (\x+0.5, \y+0.5) {\grp}; \fi \else \node[white] at (\x+0.5, \y+0.5) {\grp}; \fi \fi \else \node[black] at (\x+0.5, \y+0.5) {\grp}; \fi } }

        % 4. Dashed Kernel Box (3x3) - shifted right by 1
        \draw[dashed, ultra thick] (4,3) rectangle (7,6);
        \draw[thick] (0,0) rectangle (8,8);
    \end{tikzpicture}
    \caption{Grouped ($N{=}2$)}
    \label{fig:grouping:grouped}
\end{subfigure}
    \caption{Illustration of Context Approximation mode. This mode processes multiple
    diagonals simultaneously, using mean-fill approximation for unavailable
    context pixels.}
    \label{fig:grouping}
  \end{figure}
  % --- in preamble (recommended) ---

% \sisetup{detect-weight=true, detect-family=true}

% --- table (mock) ---
\begin{table*}
    [t]
    \centering
    \caption{Runtime and rate--distortion impact. Measured time includes whole
    network and entropy coding. Speedup is computed as a summation of encode--decode
    with regards to its autoregressive implementation in
    \cite{begaintCompressAIPyTorchLibrary2020}. BD-rate\cite{bjontegaard} is
    calculated with VVC-Intra (VTM 23.0) as the reference. $^{*}$Pre-trained
    weights or evaluated results were not publicly available from the authors.}
    \vspace{-2mm}
    \label{tab:runtime_bd}
    \setlength{\tabcolsep}{3.2pt}
    \renewcommand{\arraystretch}{1.12}

    % If you prefer adjustbox:
    % \begin{adjustbox}{max width=\textwidth}
    \resizebox{\textwidth}{!}{%
    \begin{tabular}{l|cccc|cccc|cccc}
        \toprule \multirow{2}{*}{Model}                          & \multicolumn{4}{c|}{Kodak} & \multicolumn{4}{c|}{Tecnick} & \multicolumn{4}{c}{CLIC2020-valid} \\
        \cmidrule(lr){2-5}\cmidrule(lr){6-9}\cmidrule(lr){10-13} & \makecell{Enc\\(s/img)}    & \makecell{Dec\\(s/img)}      & \makecell{Speedup\\(times)}       & \makecell{BD-rate\\(\%)} & \makecell{Enc\\(s/img)} & \makecell{Dec\\(s/img)} & \makecell{Speedup \\ (times)} & \makecell{BD-rate\\(\%)} & \makecell{Enc\\(s/img)} & \makecell{Dec\\(s/img)} & \makecell{Speedup\\(times)} & \makecell{BD-rate\\(\%)} \\
        \midrule                                                  % ---------------- Cheng2020 block ----------------
        Cheng2020 \cite{cheng2020}                               & 2.811                      & 5.937                        & 1.00                              & 5.31                     & 9.344                   & 21.219                  & 1.00                          & 6.92                     & 15.512                  & 34.073                  & 1.00                        & -2.33                    \\
        Cheng2020 + WF (ours)                                    & 0.1606                     & 0.4629                       & 14.0                              & 5.31                     & 0.2811                  & 0.9202                  & 25.4                          & 6.92                     & 0.3976                  & 1.1152                  & 32.8                        & -2.33                    \\
        Cheng2020 + WF + TrANS (ours)                            & 0.1285                     & 0.1390                       & 32.7                              & 5.31                     & 0.3004                  & 0.3260                  & 48.8                          & 6.92                     & 0.4177                  & 0.4539                  & 56.9                        & -2.33                    \\
        \midrule                                                  % ---------------- MBT2018 block ----------------
        MBT2018 \cite{minnen2018}                                & 2.5491                     & 5.6723                       & 1.00                              & 11.26                    & 9.256                   & 21.045                  & 1.00                          & 12.69                    & 14.944                  & 33.409                  & 1.00                        & 0.36                     \\
        MBT2018 + WF (ours)                                      & 0.1223                     & 0.4280                       & 14.9                              & 11.26                    & 0.2915                  & 0.9232                  & 24.9                          & 12.69                    & 0.3898                  & 1.1108                  & 32.2                        & 0.36                     \\
        MBT2018 + WF + TrANS (ours)                              & 0.1104                     & 0.1197                       & 35.7                              & 11.26                    & 0.2478                  & 0.2527                  & 60.5                          & 12.69                    & 0.3432                  & 0.3297                  & 71.9                        & 0.36                     \\
        \midrule                                                  % ---------------- InvCompress block ----------------
        InvCompress \cite{xie_enhanced_inverible_2021}           & 2.5268                     & 5.6637                       & 1.00                              & -1.24                    & 9.418                   & 21.208                  & 1.00                          & -0.99                    & 16.519                  & 35.047                  & 1.00                        & -9.25                    \\
        InvCompress + WF (ours)                                  & 0.1504                     & 0.4607                       & 13.4                              & -1.24                    & 0.3730                  & 1.0064                  & 22.2                          & -0.99                    & 0.5848                  & 1.2829                  & 27.6                        & -9.25                    \\
        InvCompress + WF + TrANS (ours)                          & 0.1457                     & 0.1607                       & 26.7                              & -1.24                    & 0.3536                  & 0.3807                  & 41.7                          & -0.99                    & 0.5046                  & 0.4953                  & 51.6                        & -9.25                    \\
        \midrule                                                  % ---------------- Non-AR / reference models ----------------
        Checkerboard\cite{heCheckerboardContextModel2021b}       & 0.0878                     & 0.0656                       & --                                & 8.51                     & 0.2575                  & 0.2002                  & --                            & -- $^{*}$                & 0.4036                  & 0.3085                  & --                          & -- $^{*}$                \\
        ELIC\cite{heELICEfficientLearned2022}                    & 0.2069                     & 0.1288                       & --                                & -7.19                    & 0.4440                  & 0.2843                  & --                            & -7.75                    & 0.6574                  & 0.4091                  & --                          & -12.88                   \\
        Minnen2020\cite{minnen2020}                              & 0.1773                     & 0.1069                       & --                                & 1.30                     & 0.3813                  & 0.2368                  & --                            & -1.63                    & 0.5469                  & 0.3403                  & --                          & -- $^{*}$                \\
        \bottomrule
    \end{tabular}%
    }
    % \end{adjustbox}
    \vspace{-5mm}
\end{table*}
  \subsection{Wavefront Implementation}
  \label{sec:wavefront_impl}

  To fully utilize GPU parallelism, we develop progressive optimizations to achieve
  practical speedup. Each optimization addresses a specific bottleneck
  identified in the previous implementation.

  \textit{Baseline (per-pixel loop).} A direct implementation iterates over
  pixels in wavefront order $\tau(i,j)=\lambda i + j$. For each diagonal $d = \lambda
  i + j$, we enumerate all valid $(i, j)$ pairs satisfying the constraint and
  process them sequentially.
  % While this schedule minimizes the number of
  % steps to $O(\lambda H + W)$ compared to raster-scan's $O(HW)$, per-pixel kernel
  % launches and Python loop overhead dominate runtime.

  \textit{Batched wavefront.} Instead of searching the symbols in the wavefront each
  time, we can first group all pixels with the same $\tau$ value into a
  wavefront $S_{s}$ and evaluate their context models in a single batched call. For
  a wavefront containing $n$ pixels at positions $(h_{1},w_{1}), \ldots , (h_{n},
  w_{n})$, we extract their $k \times k$ neighborhood patches from the padded
  latent tensor $\hat{\mathbf{Y}}$ and stack them into
  $\mathbf{P}_{s}\in \mathbb{R}^{n \times C \times k \times k}$. The masked convolution
  is then applied once:
  $\mathbf{C}_{s}= \text{Conv2D}(\mathbf{P}_{s}, \mathbf{W}\odot \mathbf{M})$, producing
  context features for all $n$ pixels simultaneously.

  \textit{Selective im2col.} Naive patch extraction using unfolding materializes
  all $HW$ sliding windows into a tensor of size $C \times k^{2}\times HW$ from
  which we then index the wavefront pixels. This full im2col is wasteful when $|S
  _{s}| \ll HW$. Instead, for each wavefront, we compute the flat indices of
  only the required $k \times k$ neighborhoods, then gather patches using PyTorch's
  advanced indexing. This selective approach avoids materializing the full
  tensor and reduces memory overhead.

  \textit{Tensorized rANS.} Once wavefront overhead is reduced, entropy coding
  becomes the bottleneck. Recent learned compression models adopt range variant of Asymmetric Neural Systems (rANS)\cite{duda2009asymmetricnumeralsystems}, and standard rANS interfaces (as in \cite{begaintCompressAIPyTorchLibrary2020}) accept Python lists,
  requiring data conversion that triggers device-to-host synchronization and
  incurs per-symbol Python overhead. We implemented a C++ interface that performs rANS encoding/decoding
  by directly pointing elements in PyTorch Tensors without explicit data conversion.

  \subsection{Context Approximation for Further Speedup}
  \label{sec:context_approx} To further improve throughput, we introduce an
  optional \emph{context approximation} mode that increases parallelism by relaxing
  strict autoregressive dependencies. As illustrated in Fig.~\ref{fig:grouping},
  we process $N$ consecutive wavefronts \emph{as a group}. Within a group, some pixels
  would normally depend on other pixels that belong to later wavefronts in the same
  group; thus, strict AR causality is not fully satisfied. Our idea is to
  approximate the unavailable context values deterministically so that both
  encoder and decoder follow the same computation and remain synchronized.

  Before computing the context for a grouped wavefront, we temporarily assign provisional
  values to the yet-unavailable pixels. Specifically, for each pixel we fill missing
  context entries by the mean of already-known reconstructed latents in its local
  neighborhood; if no known samples exist, we use a fixed constant (zero in our
  experiments). This produces an approximate context tensor that allows the
  model to evaluate $(\mu,\sigma)$ for all pixels in the grouped wavefront in a
  single batched pass. Note that this approximation results in a suboptimal probability
  modeling and higher bitrate; however, preserves image quality as the transform
  networks are identical.

  % The group size $N$ provides a simple knob to trade off runtime and compression
  % efficiency. Larger $N$ reduces the number of wavefront iterations and correspondingly
  % reduces the number of context-convolution invocations, leading to higher
  % throughput. However, because the predicted distributions are computed from
  % approximate contexts, entropy estimates become less accurate and the bitrate increases,
  % resulting in a BD-rate penalty. In Sec.~IV, we quantify this trade-off and show
  % that a single trained model can operate at multiple speed points without retraining.
  \section{Experiment}
  \label{sec:experiment}
  \subsection{Experiment Setup}
  The software environment is Python 3.10, Pytorch 2.9.0, CUDA 11.8. We used three datasets to evaluate the codecs including Kodak\cite{kodak},
  Tecnick\cite{asuni2014testimages} and CLIC2020-Valid\cite{CLIC2020}. The Kodak
  dataset contains 24 images with a resolution of $768\times512$. The Tecnick
  dataset contains 100 images with a resolution of $1200\times1200$. The CLIC2020
  validation dataset contains 41 images with various resolutions of which 65\%
  are over Full-HD resolution. Unless otherwise mentioned, the inference time is measured on a single NVIDIA
  GeForce RTX 3080 12 GB GPU and an Intel Core i7-8700K CPU. Batch size is set
  to one. We adopt the selective im2col and tensorized
  rANS as a standard baseline.

  \subsection{Implementation}
  We implement the proposed wavefront parallelization method in PyTorch. For Cheng2020\cite{cheng2020}
  and MBT2018\cite{minnen2018}, standard pre-trained weights from CompressAI\cite{begaintCompressAIPyTorchLibrary2020}
  are used. For InvCompress\cite{xie_enhanced_inverible_2021}, official pre-trained
  weights are used. For other methods, we used CompressAI implementations or re-implemented
  them based on the official papers. %The source code of re-implemented models will be publicly available.

  % --- in preamble (recommended) ---
% \usepackage{booktabs}
% \usepackage{makecell}

\begin{table}[t]
    \centering
    \caption{Speed--RD trade-off of context approximation for Cheng2020 (quality
    = 3) on Kodak. Speedups were measured relatively to the raster baseline.
    $N$ denotes the group size.}
    \label{tab:approx_kodak_cheng2020}
    \setlength{\tabcolsep}{6pt}
    \renewcommand{\arraystretch}{1.15}

    % \begin{tabular}{lcc}
    %     \toprule Method              & \makecell{PSNR\\(dB)} & \makecell{Speedup\\(Enc/Dec)} & \makecell{\\(\%)} \\
    %     \midrule Raster Baseline     & 1.00/1.00             & 0.00                           \\
    %     \midrule Staggered Wavefront & 21.87/42.71           & 0.00                           \\
    %     Approx. Group=2              & 31.71/59.37           & +9.33                          \\
    %     Approx. Group=3              & 36.96/73.04           & +14.16                         \\
    %     Approx. Group=4              & 52.05/103.79          & +27.00                         \\
    %     \bottomrule
    % \end{tabular}
    \begin{tabular}{lrrrr}
        \toprule Variant & \makecell{PSNR\\(dB)} & \makecell{BPP\\} & \makecell{$\Delta$ BPP \\ (\%)} & \makecell{Speedup \\ (Enc/Dec)} \\
        \midrule Raster  & 31.317                & 0.269            & +0.00                           & 1.00/1.00                       \\
        Wavefront        & 31.327                & 0.269            & +0.00                           & 21.87/42.71                     \\
        $N = 2$ approx.  & 31.331                & 0.295            & +9.67                           & 31.71/59.37                     \\
        $N = 3$ approx.  & 31.338                & 0.309            & +14.9                           & 36.96/73.04                     \\
        $N = 4$ approx.  & 31.346                & 0.344            & +27.9                           & 52.1/103.8                      \\
        \bottomrule
    \end{tabular}
    \vspace{-5mm}
\end{table}
  % --- table combining env2 and env3 ---
\begin{table*}
    [t]
    \centering
    \caption{Runtime comparison across two additional experimental settings. Measured time includes
    whole network and entropy coding/decoding.}
    \label{tab:runtime_combined}
    \setlength{\tabcolsep}{3.0pt}
    \renewcommand{\arraystretch}{1.12}

    \resizebox{\textwidth}{!}{%
    \begin{tabular}{l|cc|cc|cc||cc|cc|cc}
        \toprule \multirow{3}{*}{Model}                                                                                  & \multicolumn{6}{c||}{Setting 1 (RTX 2080 GPU)} & \multicolumn{6}{c}{Setting 2 (RTX 5000 PRO Blackwell GPU)} \\
        \cmidrule(lr){2-7}\cmidrule(lr){8-13}                                                                            & \multicolumn{2}{c|}{Kodak}                     & \multicolumn{2}{c|}{Tecnick}                              & \multicolumn{2}{c||}{CLIC2020-valid} & \multicolumn{2}{c|}{Kodak} & \multicolumn{2}{c|}{Tecnick} & \multicolumn{2}{c}{CLIC2020-valid} \\
        \cmidrule(lr){2-3}\cmidrule(lr){4-5}\cmidrule(lr){6-7}\cmidrule(lr){8-9}\cmidrule(lr){10-11}\cmidrule(lr){12-13} & \makecell{Enc\\(s/img)}                        & \makecell{Dec\\(s/img)}                                   & \makecell{Enc\\(s/img)}              & \makecell{Dec\\(s/img)}    & \makecell{Enc\\(s/img)}      & \makecell{Dec\\(s/img)}           & \makecell{Enc\\(s/img)} & \makecell{Dec\\(s/img)} & \makecell{Enc\\(s/img)} & \makecell{Dec\\(s/img)} & \makecell{Enc\\(s/img)} & \makecell{Dec\\(s/img)} \\
        \midrule                                                                                                          % ---------------- Cheng2020 block ----------------
        Cheng2020 (baseline)                                                                                             & 3.260                                          & 5.636                                                     & 11.733                               & 20.778                     & 18.409                       & 32.749                            & 1.418                   & 2.222                   & 5.092                   & 8.259                   & 7.941                   & 12.952                  \\
        Cheng2020 + TrANS                                                                                                & 3.219                                          & 3.495                                                     & 11.692                               & 12.800                     & 18.208                       & 20.033                            & 1.392                   & 2.007                   & 4.943                   & 7.764                   & 7.747                   & 12.337                  \\
        Cheng2020 + WF                                                                                                   & 0.152                                          & 0.372                                                     & 0.320                                & 0.794                      & 0.410                        & 0.941                             & 0.088                   & 0.151                   & 0.179                   & 0.341                   & 0.224                   & 0.416                   \\
        Cheng2020 + WF + TrANS                                                                                           & 0.138                                          & 0.148                                                     & 0.291                                & 0.311                      & 0.356                        & 0.373                             & 0.069                   & 0.088                   & 0.156                   & 0.230                   & 0.182                   & 0.267                   \\
        \midrule                                                                                                          % ---------------- MBT2018 block ----------------
        MBT2018 (baseline)                                                                                               & 3.260                                          & 5.656                                                     & 11.808                               & 20.754                     & 18.521                       & 32.503                            & 1.402                   & 2.269                   & 5.128                   & 8.282                   & 8.057                   & 12.899                  \\
        MBT2018 + TrANS                                                                                                  & 3.229                                          & 3.485                                                     & 11.818                               & 12.817                     & 18.409                       & 20.223                            & 1.363                   & 2.155                   & 5.010                   & 7.821                   & 7.761                   & 11.124                  \\
        MBT2018 + WF                                                                                                     & 0.151                                          & 0.386                                                     & 0.323                                & 0.827                      & 0.410                        & 0.993                             & 0.088                   & 0.151                   & 0.187                   & 0.346                   & 0.224                   & 0.398                   \\
        MBT2018 + WF + TrANS                                                                                             & 0.138                                          & 0.153                                                     & 0.281                                & 0.323                      & 0.344                        & 0.380                             & 0.076                   & 0.096                   & 0.157                   & 0.213                   & 0.183                   & 0.242                   \\
        \midrule                                                                                                          % ---------------- InvCompress block ----------------
        InvCompress (baseline)                                                                                           & 3.175                                          & 5.580                                                     & 11.817                               & 20.660                     & 18.739                       & 32.803                            & 1.398                   & 2.238                   & 5.204                   & 8.393                   & 8.053                   & 12.944                  \\
        InvCompress + TrANS                                                                                              & 3.138                                          & 3.414                                                     & 11.830                               & 12.876                     & 18.484                       & 20.116                            & 1.329                   & 2.026                   & 5.087                   & 7.777                   & 7.762                   & 12.818                  \\
        InvCompress + WF                                                                                                 & 0.174                                          & 0.405                                                     & 0.434                                & 0.893                      & 0.581                        & 1.107                             & 0.100                   & 0.174                   & 0.295                   & 0.456                   & 0.351                   & 0.534                   \\
        InvCompress + WF + TrANS                                                                                         & 0.164                                          & 0.178                                                     & 0.403                                & 0.415                      & 0.538                        & 0.539                             & 0.091                   & 0.119                   & 0.263                   & 0.334                   & 0.303                   & 0.394                   \\
        \bottomrule
    \end{tabular}%
    }
\end{table*}
  \section{Discussion}
  \label{sec:discussion}
  \subsection{Performance}
  As shown in Table~\ref{tab:runtime_bd} and Fig~\ref{fig:bdrate_time_kodak}, the proposed wavefront scheduling
    substantially reduces the inference time of spatial autoregressive codecs
    while preserving their RD performance. On Kodak, WF+TrANS achieves a
    26.7--35.7$\times$ speedup over the corresponding raster-scan baselines,
    and the gain further increases on larger images, reaching up to
    71.9$\times$ on CLIC2020-valid. The speedup is consistent across three
    different spatial-AR backbones, indicating that the proposed scheduling is
    not specific to a particular model architecture. 
    
    Table~\ref{tab:acceleration} provides a comparison between other acceleration methods. Checkerboard decoding\cite{heCheckerboardContextModel2021b} provides the shortest runtime, however, its architectural simplification degrades RD performance and
    requires retraining. The BD-rate increases from 5.31\% to
    8.51\% for Cheng2020\cite{cheng2020} and from 11.26\% to 16.71\% for MBT2018\cite{minnen2018}. In contrast,
    wavefront scheduling preserves the original BD-rate of the raster-scan
    models while reducing the decoding time from 5.64s to 0.148s for
    Cheng2020 and from 5.66s to 0.166s for MBT2018. Thus, although
    checkerboard remains faster in absolute runtime, wavefront provides a
    more favorable speed--RD trade-off for pre-trained spatial AR codecs by
    achieving deployable acceleration without retraining or sacrificing
    compression efficiency.

  To demonstrate that our methods are effective regardless of the hardware environment, we conducted experiments in other hardware settings together with detailed description.  \textbf{Setting 1. (Low)}An Intel Core i7-7700 CPU with RTX 2080 SUPER GPU is deployed. \textbf{Setting 2. (High)}An Intel Core Ultra 9 285K with RTX 5000 PRO Blackwell GPU is deployed. As shown in Table~\ref{tab:runtime_combined}, wavefront parallelization (WF) alone achieves up to 97\% reduction in decoding time compared to the raster-scan baseline (e.g., 5.636s$\to$0.148s for Cheng2020 on Kodak in Setting 1). When combined with Tensorized rANS, decoding speedup reaches approximately 38$\times$. Note that applying TrANS to the sequential raster-scan baseline also improves speed (5.636s$\to$3.495s,$\sim$1.6$\times$), but the gain is substantially smaller than in the WF setting. This observation confirms that once wavefront parallelization increases the degree of parallelism, the entropy
  coding becomes the dominant bottleneck, making TrANS essential for further
  acceleration.

  \subsection{Exactness and Reproducibility}
  
  In practice, changing the GPU execution order introduces minor floating-point
  rounding differences, and the wavefront-based rANS update order causes slight numerical
  variations. However, with our wavefront implementation of Cheng2020, the
  maximum relative error in PSNR and BPP compared to the raster baseline is less
  than 0.08\% across all Kodak images, confirming that the proposed scheduling
  preserves the original model's compression rate performance.
  \subsection{Ablation of implementations}
  \label{sec:ablation}
We conducted an ablation study to validate our optimizations (Table~\ref{tab:architecture_ablation}).
  First, to ensure a fair comparison, we applied Tensorized rANS to the standard
  raster-scan method. This yields a $2.49$s decode time; while faster ($\approx 2
  .4\times$) than the standard baseline, the serial bottleneck remains.

  In contrast, our proposed Wavefront method combined with Tensorized rANS achieves
  a decoding time of $0.14$s, representing a massive $\approx18\times$ speedup even
  over the \textit{optimized} raster baseline (and $>40\times$ over the standard
  baseline). Notably, the speedup gain from Tensorized rANS is more pronounced
  in the Wavefront setting ($\approx 3.3\times$, from $0.46$s to $0.14$s) compared
  to the Raster setting ($\approx 2.4\times$). This confirms our observation
  that once neural network inference is parallelized via Wavefront, the bottleneck
  shifts to entropy coding overhead. Ultimately, the combination of selective batching and Tensorized rANS yields the best performance. \looseness=-5
  \begin{table}[t]
\centering
\caption{Comparison of acceleration methods on Cheng2020\cite{cheng2020} and MBT2018\cite{minnen2018}.}
\label{tab:acceleration}
\begin{tabular}{llccc}
\toprule
Model & Method & Speed & BD-Rate & Training- \\
 & & (s) & (\%) & Free \\
\midrule
\multirow{3}{*}{MBT2018}
  & None(raster) & 5.66  & \textbf{11.26} &  \\
  & checkerboard & \textbf{0.016} & 16.71 &  \\
  & wavefront    & 0.166 & \textbf{11.26}  & \checkmark \\
\midrule
\multirow{3}{*}{Cheng2020}
  & None(raster) & 5.64  & \textbf{5.31} &  \\
  & checkerboard & \textbf{0.066} & 8.51 &  \\
  & wavefront    & 0.148 & \textbf{5.31} & \checkmark \\
\bottomrule
\end{tabular}
\end{table}

  % \subsection{Bottleneck analysis}
  % Although our wavefront schedule exposes parallelism, the achieved speedup depends
  % on the end-to-end bottleneck. Fig.~\ref{fig:bdrate_time_kodak} reports a runtime
  % breakdown for the raster baseline, exact wavefront, and wavefront+opt implementations.
  % We observe that the dominant cost is the context convolution (\texttt{context\_prediction}),
  % accounting for 60\%--70\% of the total runtime, while entropy-parameter
  % prediction and rANS coding contribute 15\% and 20\%, respectively. This
  % explains why further reductions in neighborhood materialization (e.g., full vs.\ selective
  % im2col) provide limited gains in our setting: the performance is primarily
  % bounded by the number of convolution invocations and kernel-launch overhead. Our
  % \texttt{WF+Opt} implementation improves throughput by 2.3\texttimes{} over
  % \texttt{WF} mainly by reducing host-side overhead (e.g., eliminating Python
  % list conversion in entropy coding) and improving batching efficiency.
  \begin{table}[t]
    \caption{Ablation of GPU-friendly implementation of Wavefront, evaluated with
    Cheng2020\cite{cheng2020} (quality = 3) over Kodak\cite{kodak} dataset.
    % Naive wavefront
    % uses a diagonal loop with per-pixel context. Batched wavefront groups pixels
    % via unfold-based batching. Selective wavefront reduces unfolded tensor
    % materialization by selective batching, and the "+Tensor rANS" version combines
    % that with Tensorilized rANS implementation.
    }
    \label{tab:architecture_ablation}
    \centering
    \begin{tabular}{lccc}
        \toprule Variant       & Enc (s)         & Dec (s)         & bpp / PSNR (dB)     \\
        \midrule Original Raster-Scan  & 2.811          & 5.937          & 0.269 / 31.32 \\
        + Tensorized rANS & 2.493 & 2.694 & 0.269 / 31.32 \\ 
        \midrule Loop-based WF & 0.1815          & 0.4676          & 0.269 / 31.33 \\
        \midrule Batched WF    & 0.1394          & 0.4512          & 0.269 / 31.33 \\
        + Tensorized rANS          & 0.1702          & 0.2096          & 0.269 / 31.33 \\
        \midrule Selective WF  & 0.1636          & 0.4576          & 0.269 / 31.33 \\
        + Tensorized rANS          & \textbf{0.1237} & \textbf{0.1376} & 0.269 / 31.33 \\
        \bottomrule
    \end{tabular}
    \vspace{-4mm}
\end{table}
  \subsection{Context approximate mode}
  % We further propose an optional context-approximation mode (GroupFill-$N$) that
  % increases parallelism by processing $N$ wavefronts jointly and deterministically
  % filling unavailable context values. This mode enables a single trained model
  % to operate at multiple speed points without retraining.
  We evaluated the effectiveness of context approximation (Sec.
  \ref{sec:context_approx}). As shown in Table~\ref{tab:approx_kodak_cheng2020},
  larger $N$ yields higher speedups at the cost of increased bitrate. Image
  quality is preserved except for minor quantization errors. When we measure the
  rate-distortion with all qualities Cheng2020 and Kodak, the BD-Rate for
  $N = \lbrace 1,2,3,4 \rbrace$ are $\lbrace0.00, 9.33, 14.16, 27.00\rbrace \%$,
  respectively. We view this as a practical speed--rate tradeoff for applications
  where throughput is prioritized over compression efficiency. 
  % %ref tab:quality
  % \input{tables/table_quality}
  \subsection{Limitations}
  Firstly, changing the processing order alters the bitstream layout. The
  compressed bitstream from our method is therefore not inter-decodable with the
  raster baseline, and vice versa. Nevertheless, model weights remain compatible
  with the proposed method applied on autoregressive models, and they do not
  have to be retrained.

  Secondly, our method targets spatial autoregressive models with local dependencies.
  It does not benefit Checkerboard-style models
  \cite{heCheckerboardContextModel2021b} or channel-wise autoregressive models
  \cite{minnen2020}. However, in case the model combines channel and spatial
  autoregressive context\cite{crosschannelcontext}, our wavefront methods is
  applicable and can potentially bring further enhancement.

  % However, recent models \cite{heELICEfficientLearned2022} use the combination of channel and spatial context, and as future work, the proposed method can potentially bring further enhancement to such methods.

  \section{Conclusion}
  \label{sec:conclusion}
  We presented a training-free acceleration of spatial autoregressive context
  models for learned image compression. By modeling the spatial context dependency
  as a DAG and deriving an optimal schedule in the form of a staggered wavefront
  order, we enable safe parallel execution without modifying model parameters. Combined
  with a GPU-friendly batched implementation, our method achieves more than
  $26\times$ speedup (or $13\times$ with wavefront alone) while preserving
  compression performance. We further introduced an optional context-approximation
  mode that offers variable-speed inference from a single model by trading bitrate
  for throughput.

  \bibliographystyle{IEEEbib}
  \bibliography{refs}

% Template for ICIP-2026 paper; to be used with:
%          spconf.sty  - ICASSP/ICIP LaTeX style file, and
%          IEEEbib.bst - IEEE bibliography style file.
% --------------------------------------------------------------------------
\appendices

  \section{Proof of the Validity and Optimality of Wavefront Parallelization}
  In this section, we provide a formal proof of the validity and optimality of
  the wavefront parallelization strategy in general case.
  % Preamble (already likely in your paper):
  % \usepackage{amsmath,amssymb,amsthm}
  % \newtheorem{theorem}{Theorem}
  \begin{definition}[Dependency set for raster-scan causal masked convolution] Let
  $\Omega=\{0,\dots,I-1\}\times\{0,\dots,J-1\}\subset\mathbb{Z}^{2}$, where
  $i\in\{0,\dots,I-1\}$ is the \emph{row} index and $j\in\{0,\dots,J-1\}$ is the \emph
  {column} index. Let $R\in\mathbb{N}$. A raster-scan $(2R{+}1)\times(2R{+}1)$ causal masked convolution uses as context: (i) all pixels in the $R$ rows above within horizontal offset at most$R$, and (ii) the $R$ pixels to the left in the same row. Equivalently, define the dependence vector set
  \[
    D_{R}\;:=\; (\{1,\dots,R\}\times\{-R,\dots,R\}) \;\cup\; (\{0\}\times\{1,\dots
    ,R\}).
  \]
  For each $x\in\Omega$, the dependency graph contains an edge $x-d\to x$ for every $d\in
  D_{R}$ such that $x-d\in\Omega$. \end{definition}
  \begin{definition}[Staggered wavefront schedule]
    % force newline
    Let $\lambda := R+1$. \emph{(Note: this $\lambda$ corresponds to the shear
    factor $\lambda$ in the main paper)} Fix $\pi=(\lambda,1)$ and define the rank
    function $t:\Omega\to\mathbb{Z}$ by
    \[
      t(i,j)=\pi\cdot(i,j)=\lambda i + j.
    \]
    The \emph{staggered wavefront schedule} executes nodes in increasing order
    of $t$, and executes all nodes with the same value of $t$ in one parallel step
    (i.e., level sets $\{(i,j)\in\Omega: t(i,j)=c\}$ form the parallel batches).
  \end{definition}
  \begin{theorem}[Validity and optimality of the staggered wavefront schedule]
    Consider the dependency set $D_{R}$ from Definition~1 and assume $J\ge R+1$.
    Let $\lambda:=R+1$ as in Definition~2. Then:
    \begin{enumerate}
      \item The staggered wavefront schedule is valid for $D_{R}$.

      \item In the unit-time, infinite-processor model, the staggered wavefront schedule
        is optimal, i.e., it achieves the minimum possible makespan
        \[
          T^{\*}\;=\;\lambda I + J - \lambda \;=\; (R+1)I + J - (R+1).
        \]
    \end{enumerate}
  \end{theorem}
  \begin{proof}
    (1) \textbf{Validity.} For any $d\in D_{R}$, we have
    \[
      \pi\cdot d=
      \begin{cases}
        p,             & d=(0,p),\ p\in\{1,\dots,R\},                      \\
        \lambda r + q, & d=(r,q),\ r\in\{1,\dots,R\},\ q\in\{-R,\dots,R\}.
      \end{cases}
    \]
    In the first case, $\pi\cdot d=p\ge 1$. In the second case, the minimum is
    attained at $(r,q)=(1,-R)$, giving $\pi\cdot d \ge \lambda - R = (R+1)-R=1$.
    Hence $\pi\cdot d \ge 1>0$ for all $d\in D_{R}$. Therefore the condition of Lamport's
    hyperplane theorem holds for $t(x)=\pi\cdot x$, so executing level sets
    $\{x: t(x)=c\}$ in increasing $c$ is a valid wavefront schedule.

    (2) \textbf{Optimality.} Let $T^{\*}$ denote the optimal makespan.

    \emph{Upper bound.} Over $\Omega$, $t_{\min}=t(0,0)=0$ and $t_{\max}=t(I-1,J-
    1)=\lambda(I-1)+(J-1)$. Thus the staggered wavefront schedule completes in
    \[
      T_{\mathrm{wf}}=(t_{\max}-t_{\min})+1 = \lambda(I-1)+J = \lambda I + J - \lambda
    \]
    steps, implying $T^{\*}\le T_{\mathrm{wf}}$.

    \emph{Lower bound.} Any schedule must take at least the length of the
    longest directed path (critical path) in the dependency DAG. We construct a directed
    path of length $T_{\mathrm{wf}}$. First traverse the top row using
    $(0,1)\in D_{R}$:
    \[
      (0,0)\to(0,1)\to\cdots\to(0,J-1)\quad (\text{$J$ nodes}).
    \]
    Then, for each $i=0,\dots,I-2$, use $(1,-R)\in D_{R}$ to jump
    \[
      (i,J-1)\to(i+1,J-1-R),
    \]
    and move right by $R$ steps within row $i+1$ using $(0,1)\in D_{R}$ to reach
    $(i+1,J-1)$. Each row transition adds $(R+1)=\lambda$ nodes, yielding a path
    of length
    \[
      J+\lambda(I-1)=\lambda I + J - \lambda = T_{\mathrm{wf}}.
    \]
    Hence $T^{\*}\ge T_{\mathrm{wf}}$. Combining with the upper bound gives
    $T^{\*}=T_{\mathrm{wf}}$.
  \end{proof}

  \vfill
  \pagebreak

  \section{Torch-like Pseudocode for Wavefront Implementation}% --- V1 Code ---
  \begin{lstlisting}[caption={Wavefront Decoding with Naive Loop-Based Processing}, label={lst:v1}]
params = hyper_network(bitstream_hyper)
# means and scales for each pixel from hyper-prior
# of shape (N, 2*C, H, W)
y_hat = torch.zeros(N, C, H, W)

# Pixel-wise processing
for h, w in get_diag_coords(H, W):
    # 1. Individual crop (Fragmented access)
    y_crop = y_hat[:, :, h:h+5, w:w+5]
    
    # 2. Compute & Update
    ctx = F.conv2d(y_crop, weight, bias=b)
    y_val = decode(params[..., h, w], ctx)
    y_hat[:, :, h+2, w+2] = y_val
\end{lstlisting}

  % --- V2 Code ---
  \begin{lstlisting}[caption={Wavefront Decoding with Unfold-based Batching}, label={lst:v2}]
params = hyper_network(bitstream_hyper)
params_flat = params.view(N, C, -1)
y_hat = torch.zeros(N, C, H, W)
# flat view must be taken AFTER padding/zeroing to point to correct memory
y_hat_flat = y_hat.view(N, C, -1) 

# Batching with F.unfold
for batch_indices in get_diag_indices(H, W):
    # 1. Unfold entire image (High Memory Cost)
    #    Inefficient: Extracts patches for ALL pixels every time
    patches_all = F.unfold(y_hat, 5) 
    patches = patches_all[..., batch_indices]
    
    # 2. Batch Compute
    ctx = F.conv2d(patches, weight, bias=b)
    
    # 3. Gather Params for current batch (Fixed)
    #    Must slice params to match ctx dimensions
    p_batch = params_flat[..., batch_indices].unsqueeze(-1).unsqueeze(-1)
    
    # 4. Decode & Update
    y_vals = decode(p_batch, ctx)
    y_hat_flat.scatter_(1, batch_indices, y_vals)
\end{lstlisting}
  %move to right column
  \vfill
  % --- V3 Code ---
  \begin{lstlisting}[caption={Wavefront Decoding with Selective Im2Col}, label={lst:v3}]
params = hyper_network(bitstream_hyper)
params_flat = params.view(N, C, -1)
y_hat_flat = torch.zeros(N, C, H*W) 
# Selective Im2Col
offset = get_kernel_offsets(kernel_size=5)
for centers in get_diag_centers(H, W):
    # 1. Gather Patches (Memory Efficient)
    #    Calculate flat indices for 5x5 neighbors
    idx = centers.unsqueeze(1) + offset
    patches = y_hat_flat[:, idx].view(N, C, 5, 5)

    # 2. Compute Context from Patches
    ctx = F.conv2d(patches, masked_weight, bias=b)

    # 3. Fuse with Hyper-priors (params)
    #    Gather params at the same center positions
    p = params_flat[:, centers].unsqueeze(-1).unsqueeze(-1)

    y_vals = decode(p, ctx)

    # 5. Scatter Update
    y_hat_flat.scatter_(1, centers, y_vals)
y_hat = y_hat_flat.view(N, C, H, W)
\end{lstlisting}

  \section{Settings of VVC-intra}

  We evaluated VVC reference coding using VTM-23.0. We found that using standard
  CompressAI's benchmark results in inferior performance due to color space
  settings, therefore we conducted experiments directly with VTM-23.0 and ffmpeg.

  The encoder was configured with the default configuration file. Input images
  were 8-bit RGB. Prior to encoding, each image was converted to 8-bit YUV444
  raw video using ffmpeg. The reconstructed YUV was converted back to RGB (ffmpeg),
  and PSNR was computed in the RGB domain against the original image. QP were set
  to 22, 27, 32, 37, 42 and 47.

  \begin{enumerate}
    \item \textbf{RGB $\rightarrow$ YUV444 conversion (ffmpeg)}
      {\small \begin{verbatim}
ffmpeg -y -i input.png -pix_fmt 
yuv444p -f rawvideo input.yuv
\end{verbatim} }

    \item \textbf{Encoding with VTM (EncoderApp)}
      {\small \begin{verbatim}
EncoderApp -c encoder_intra_vtm.cfg \\
 -i input.yuv -b output.bin \\
  -o recon.yuv -f 1 -fr 2 \\
  -wdt W -hgt H  -q QP \\
  --InputBitDepth=8 \\
  --InputChromaFormat=444 \\
  --OutputBitDepth=8 \\
  --OutputBitDepthC=8 \\
  --ConformanceWindowMode=1
\end{verbatim} }

    \item \textbf{YUV $\rightarrow$ RGB reconstruction (ffmpeg)}
      {\small \begin{verbatim}
ffmpeg -y -f rawvideo -pix_fmt yuv444p  
-s W*H -i recon.yuv recon.png
\end{verbatim} }
  \end{enumerate}
\end{document}